\documentclass[submission,copyright,creativecommons]{eptcs}
\providecommand{\event}{GandALF 2025} 

\usepackage{iftex}

\ifpdf
  \usepackage{underscore}         
  \usepackage[T1]{fontenc}        
\else
  \usepackage{breakurl}           
\fi

\usepackage{amstext,bm}
\usepackage{amsthm}
\usepackage{amsopn}
\usepackage{enumerate}
\usepackage{tikz}
\usetikzlibrary{trees}
\usetikzlibrary{shapes,arrows}
\usetikzlibrary{intersections}
\usepackage[utf8]{inputenc}
\usepackage{amsmath}
\usepackage{amssymb,amsfonts}
\usepackage[T1]{fontenc}
\usepackage{bm}
\usepackage{graphicx}
\usepackage{caption,newfloat}
\usepackage{hyperref}
\usepackage{wrapfig}
\usepackage{blindtext}
\usepackage{array}
\usepackage{makecell}
\usepackage{mathtools}
\usepackage{float}
\usepackage{algpseudocode}
\usepackage[linesnumbered,lined,boxed,commentsnumbered,noend]{algorithm2e}
\SetKwFunction{process}{Process}
\SetKwFunction{gcd}{MLB}
\SetKwRepeat{Do}{do}{while}%
\usepackage{stmaryrd}
\usepackage{thmtools}
\usepackage{thm-restate}
\usepackage{cleveref}
\usepackage{bussproofs}
\usetikzlibrary{external}
\usepackage{colortbl}
\usepackage{hhline}
\usepackage{multirow}
\usepackage{booktabs}
\newcolumntype{P}[1]{>{\centering\arraybackslash}p{#1}}

\newcommand{\classize}[1]{{\ensuremath{\mathsf{#1}}}\xspace}
\newcommand{\pspace}{\classize{PSPACE}}
\newcommand{\expc}{\classize{EXP}}
\newcommand{\cP}{\classize{P}}
\newcommand{\NP}{\classize{NP}}
\newcommand{\conp}{\classize{coNP}}
\newcommand{\us}{\classize{US}}

\newcommand{\dpc}{\classize{DP}}

\newcommand{\ap}{\classize{AP}}

\newcommand{\logicize}[1]{{\ensuremath{\mathbf{#1}}}\xspace}

\newcommand{\hml}{\logicize{HML}}

\newcommand{\mathL}{\ensuremath{\mathcal{L}}\xspace}
\newcommand{\mathS}{\ensuremath{\mathcal{S}}\xspace}

\newcommand{\true}{\ensuremath{\mathbf{tt}}\xspace}
\newcommand{\ff}{\ensuremath{\mathbf{ff}}\xspace}

\newcommand{\curle}{\lesssim}
\newcommand{\notcurle}{\not\curle}

\newcommand{\act}{\ensuremath{\mathtt{Act}}\xspace}
\newcommand{\proc}{\ensuremath{\mathtt{Proc}}\xspace}

\newcommand{\sub}{\ensuremath{\mathrm{Sub}}}

\newcommand{\expphi}{\ensuremath{\mathrm{Exp}(\varphi)}\xspace}
\newcommand{\md}{\ensuremath{\mathrm{md}}}

\newcommand{\depth}{\ensuremath{\mathrm{depth}}}

\newcommand{\grcd}{\ensuremath{\mathrm{mlb}}}

\newcommand{\reach}{\ensuremath{\mathrm{reach}}}

\newcommand{\trim}{\ensuremath{\mathrm{trim}}}

\newcommand{\inlabeli}{\ensuremath{L_i^\mathtt{in}}\xspace}
\newcommand{\finlabeli}{\ensuremath{L_i^\mathtt{fin}}\xspace}
\newcommand{\inlabelone}{\ensuremath{L_1^\mathtt{in}}\xspace}
\newcommand{\finlabelone}{\ensuremath{L_1^\mathtt{fin}}\xspace}

\newcommand{\inlabelthree}{\ensuremath{L_3^\mathtt{in}}\xspace}
\newcommand{\finlabelthree}{\ensuremath{L_3^\mathtt{fin}}\xspace}
\newcommand{\starlabelthree}{\ensuremath{L_3^{\mathrm{sub}}}\xspace}
\newcommand{\wit}{\ensuremath{\mathrm{wit}_{\not\curle_S}}\xspace}
\newcommand{\witproc}{\ensuremath{\mathrm{witproc}_{\not\curle_S}}\xspace}
\newcommand{\cop}{\ensuremath{\mathrm{copy}}\xspace}
\newcommand{\simequiv}{\ensuremath{\mathtt{char1se}}\xspace}

\newcommand{\nsimeq}{\ensuremath{\mathtt{charnse}}\xspace}
\newcommand{\isimeq}{\ensuremath{\mathtt{charise}}\xspace}
\newcommand{\nsimpre}{\ensuremath{\mathtt{primensp}}\xspace}
\newcommand{\nmosimeq}{\ensuremath{\mathrm{char(n-1)se}}\xspace}

\newcommand{\simab}{\ensuremath{\mathtt{Sim_{A,B}}}\xspace}
\newcommand{\nmosimab}{\ensuremath{\mathtt{Sim^{n-1}_{A,B}}}\xspace}
\newcommand{\nsimab}{\ensuremath{\mathtt{Sim^{n}_{A,B}}}\xspace}
\newcommand{\isimab}{\ensuremath{\mathtt{Sim^{i}_{A,B}}}\xspace}
\newcommand\myarrowone{\mathrel{\stackrel{\makebox[0pt]{\mbox{\normalfont \scriptsize $1$}}}{\longrightarrow}}}
\newcommand\myarrowk{\mathrel{\stackrel{\makebox[0pt]{\mbox{\normalfont \scriptsize $k$}}}{\longrightarrow}}}
\newcommand\myarrowi{\mathrel{\stackrel{\makebox[0pt]{\mbox{\normalfont \scriptsize $i$}}}{\longrightarrow}}}
\newcommand\myarrowa{\mathrel{\stackrel{\makebox[0pt]{\mbox{\normalfont \scriptsize $a$}}}{\longrightarrow}}}
\newcommand\myarrowb{\mathrel{\stackrel{\makebox[0pt]{\mbox{\normalfont \scriptsize $b$}}}{\longrightarrow}}}

\newcommand\myarrowasubi{\mathrel{\stackrel{\makebox[0pt]{\mbox{\normalfont \scriptsize $a_i$}}}{\longrightarrow}}}
\newcommand\myarrowasubone{\mathrel{\stackrel{\makebox[0pt]{\mbox{\normalfont \scriptsize $a_{i_1}$}}}{\longrightarrow}}}
\newcommand\myarrowasubtwo{\mathrel{\stackrel{\makebox[0pt]{\mbox{\normalfont \scriptsize $a_{i_2}$}}}{\longrightarrow}}}
\newcommand\myarrowasubm{\mathrel{\stackrel{\makebox[0pt]{\mbox{\normalfont \scriptsize $a_{i_m}$}}}{\longrightarrow}}}
\newcommand\myarrowasubk{\mathrel{\stackrel{\makebox[0pt]{\mbox{\normalfont \scriptsize $a_{i_k}$}}}{\longrightarrow}}}
\newcommand\notmyarrowasubk{\mathrel{\stackrel{\makebox[0pt]{\mbox{\normalfont \scriptsize $a_{i_k}$}}}{\not\rightarrow}}}
\newcommand\myarrowasubkplusone{\mathrel{\stackrel{\makebox[0pt]{\mbox{\normalfont \scriptsize $a_{i_{(k+1)}}$}}}{\longrightarrow}}}
\newcommand\myarrowasubkminusone{\mathrel{\stackrel{\makebox[0pt]{\mbox{\normalfont \scriptsize $a_{i_{(k-1)}}$}}}{\longrightarrow}}}

\newcommand\myarrowasubj{\mathrel{\stackrel{\makebox[0pt]{\mbox{\normalfont \scriptsize $a_j$}}}{\longrightarrow}}}

\newcommand\notmyarrowb{\mathrel{\stackrel{\makebox[0pt]{\mbox{\normalfont \scriptsize $b$}}}{\not\rightarrow}}}

\newcommand\notmyarrowasubi{\mathrel{\stackrel{\makebox[0pt]{\mbox{\normalfont \scriptsize $a_i$}}}{\not\rightarrow}}}

\newcommand\myarrowtau{\mathrel{\stackrel{\makebox[0pt]{\mbox{\normalfont \scriptsize $t$}}}{\longrightarrow}}}

\newcommand\notmyarrowa{\mathrel{\stackrel{\makebox[0pt]{\mbox{\normalfont \scriptsize $a$}}}{\not\rightarrow}}}

\newcommand\vartextvisiblespace[1][.5em]{%
  \makebox[#1]{%
    \kern.07em
    \vrule height.3ex
    \hrulefill
    \vrule height.3ex
    \kern.07em
  }
}

\newtheorem{theorem}{Theorem}
\newtheorem{lemma}[theorem]{Lemma} 
\newtheorem{corollary}[theorem]{Corollary} 
\newtheorem{proposition}[theorem]{Proposition}
\theoremstyle{definition}
\newtheorem{definition}{Definition}
\newtheorem{example}[theorem]{Example}
\theoremstyle{remark}
\newtheorem{remark}{Remark}
\newtheorem{claim}{Claim}

\newenvironment{claimproof}{\paragraph{\textcolor{darkgray}{\textit{Proof of claim.}}}}{\hfill$\blacktriangleleft$}

\newenvironment{proofsketch}{\paragraph{\textcolor{darkgray}{\textit{Proof sketch.}}}}{\hfill$\square$}

\title{The Complexity of Deciding Characteristic Formulae Modulo Nested Simulation (extended abstract)}
\author{Luca Aceto
\institute{Department of Computer Science\\ Reykjavik University\\ Reykjavik, Iceland}
\institute{Gran Sasso Science Institute\\ L'Aquila, Italy
}
\email{luca@ru.is}
\and
Antonis Achilleos 
\institute{Department of Computer Science\\ Reykjavik University\\ Reykjavik, Iceland}
\email{antonios@ru.is}
\and
Aggeliki Chalki
\institute{Department of Computer Science\\ Reykjavik University\\ Reykjavik, Iceland}
\email{angelikic@ru.is}
\and 
Anna Ing\'olfsd\'ottir
\institute{Department of Computer Science\\ Reykjavik University\\ Reykjavik, Iceland}
\email{annai@ru.is}
}
\def\titlerunning{The Complexity of Deciding Characteristic Formulae Modulo Nested Simulation}
\def\authorrunning{L. Aceto, A. Achilleos, A. Chalki \& A. Ing\'olfsd\'ottir}

\begin{document}
\maketitle

\begin{abstract}
This paper studies the complexity of determining whether a formula in the modal logics characterizing the nested-simulation semantics is characteristic for some process, which is equivalent to determining whether the formula is satisfiable and prime. The main results are that the problem of determining whether a formula is prime in the modal logic characterizing the 2-nested-simulation preorder is {\conp}-complete and is {\pspace}-complete in the case of the $n$-nested-simulation preorder, when $n\geq 3$. This establishes that deciding characteristic formulae for the $n$-nested simulation semantics is \pspace-complete, when $n\geq 3$. In the case of the 2-nested simulation semantics, that problem lies in the complexity class \dpc,  which consists of languages that can be expressed as the intersection of one language in \NP and of one in \conp.
\end{abstract}

\section{Introduction}\label{section:intro}

Since the pioneering work by Hennessy and Milner~\cite{HennessyM85,Milner81}, behavioural equivalences and preorders over processes have been given logical characterizations, typically using some modal logic---see, for instance,~\cite{Glabbeek01} for a survey. Such characterizations state that two processes are related by some behavioural equivalence if, and only if, they satisfy the same properties expressible in some logic. 
A formula $\varphi$ is characteristic for a process $p$ with respect to some behavioural equivalence or preorder $\curle$ if every process $q$ satisfies $\varphi$ exactly when $p \curle q$ holds. 
Thus, characteristic formulae provide a complete logical description of processes up to some notion of behavioural equivalence or preorder. As shown in, e.g.,~\cite{AcetoMFI19,AcetoILS12,BrowneCG88,GrafS86a,SteffenI94}, the logics that underlie classic modal characterization theorems for equivalences and preorders over processes allow one to express characteristic formulae. Moreover, the procedures for constructing characteristic formulae presented in those works 
reduce equivalence and preorder checking to model checking---see~\cite{CleavelandS91} for an application of this approach. 
The converse reduction from model checking problems to equivalence and preorder checking problems is only possible when the logical specifications are characteristic formulae~\cite{BoudolL92}. This raises the natural question of how to determine whether a specification, expressed as a modal formula, is characteristic for some process and of the computational complexity of that problem. 
In~\cite{AcetoMFI19,AFEIP11,BoudolL92}, it was shown that characteristic formulae coincide with those that are both \emph{satisfiable and prime}. (A formula is prime if whenever it entails a disjunction $\varphi_1 \vee \varphi_2$, then it entails $\varphi_1$ or $\varphi_2$.) So, determining whether a formula is characteristic is equivalent to checking its satisfiability and primality. 

In the setting of the modal logics that characterize bisimilarity,
the problem of checking whether a formula is characteristic 
has the same complexity as validity checking, which is \pspace-complete for Hennessy-Milner logic (\hml) and \expc-complete for its extension with fixed-point operators and the $\mu$-calculus~\cite{AcetoAFI20,Achilleos18}. In~\cite{AcetoACI25}, we studied the complexity of the problem modulo the simulation-based semantics considered by van Glabbeek in his seminal linear-time/branching-time spectrum~\cite{Glabbeek01}. In op.~cit., we identified the complexity of satisfiability and primality for fragments of \hml that characterize various relations in the spectrum, delineating the boundary between the logics for which those problems can be solved in polynomial time and the logics for which they are computationally hard. 
Both satisfiablity and primality checking are decidable in polynomial time for simulation, complete simulation, and ready simulation when the set of actions has constant size.  Computational hardness already manifests itself in ready simulation semantics~\cite{BloomIM95,LarsenS91} when the size of the action set is not a constant: in this case, the problems of checking satisfiability and primality for formulae in the logic characterizing the ready simulation preorder are \NP-complete and \conp-complete, respectively. In the presence of at least two actions, for the logic characterizing the 2-nested-simulation preorder, satisfiability and primality checking are \NP-complete and \conp-hard, respectively, while deciding whether a formula is characteristic is \us-hard~\cite{BlassG82} (that is, it is at least as hard as the problem of deciding whether a given Boolean formula has exactly one satisfying truth assignment). Moreover, all three problems---satisfiability and primality checking, and deciding characteristic formulae---are \pspace-hard in the modal logic for the 3-nested-simulation preorder~\cite{GV92}. Additionally, deciding characteristic formulae modulo the equivalence relations induced by the 2-nested and 3-nested-simulation preorders is \conp-hard and \pspace-hard, respectively.\footnote{The family of nested-simulation equivalences and preorders were introduced by Groote and Vaandrager in~\cite{GV92}, where they proved that $2$-nested simulation equivalence is the completed trace congruence induced by the operators definable by rules in pure \emph{tyft/tyxt} format and that the intersection of all the $n$-nested simulation equivalences is bisimilarityfor processes satisfying a classic and mild finiteness condition.}

The work presented in~\cite{AcetoACI25} did not provide upper bounds on the complexity of the aforementioned problems for the family of $n$-nested-simulation semantics with $n\geq 2$. In this paper, we give algorithms showing that deciding whether a formula is prime in the logics characterizing the 2-nested and $n$-nested-simulation preorders, with $n\geq 3$, is in \conp and \pspace, respectively. These results, combined with previously-known lower bounds, imply that the problem is \conp-complete and \pspace-complete, respectively (Theorems~\ref{prop:3S-char-in-pspace} and~\ref{prop:2S-primality-conpc}). Taking into account the complexity of the satisfiability problem for the respective logics, we show that deciding characteristic formulae for the $n$-nested-simulation semantics is \pspace-complete when $n\geq 3$ (Corollaries~\ref{cor:n-char-equiv} and~\ref{cor:3s-char-pspacec}) and is in the complexity class \dpc for the 2-nested-simulation semantics (Corollary~\ref{cor:2S-char-dp}).

Our algorithms 
are based on two families of two-player, zero-sum games that are initiated on a modal formula $\varphi$ (Section~\ref{section:deciding-char-nS}). In the first type of game, defined for every $n\geq 1$, player $B$ constructs two structures corresponding to two arbitrary processes that satisfy $\varphi$ and player $A$ wins if she manages to show that the first process is $n$-nested-simulated by the second one. These games are used to determine whether all processes satisfying 
$\varphi$ are equivalent modulo the $n$-nested-simulation equivalence, and thus whether $\varphi$ 
 is characteristic modulo the $n$-nested-simulation equivalence relation. The second class of games, introduced for $n\geq 3$, is designed so that when a game is initiated on a satisfiable formula $\varphi$, a winning strategy for player $A$ is equivalent to $\varphi$ being prime, and hence characteristic, modulo the $n$-nested-simulation preorder. By adapting the algorithms that arise from these games, we provide \conp algorithms for the case of the 2-nested-simulation preorder and equivalence relation in Section~\ref{section:deciding-char-2S}.

The proofs of all the results announced in this paper may be found in the full version~\cite{AcetoACI2025full}.

\section{Preliminaries}\label{section:prelim}

\paragraph{Concurrency theory and logic}
In this paper, we model processes as finite, loop-free \emph{labelled transition systems} (LTS). A finite LTS is a triple $\mathS=(\proc,\act,\longrightarrow)$, where $\proc$ is a finite set of states (or processes), \act is a finite, non-empty set of actions and ${\longrightarrow}\subseteq \proc\times \act\times \proc$ is a transition relation. As usual, we use $p\myarrowa q$ instead of $(p,a,q)\in {\longrightarrow}$. For each $t\in \act^*$, we write $p\myarrowtau q$ to mean that there is a sequence of transitions labelled with $t$ starting from $p$ and ending at $q$. An LTS is \emph{loop-free} iff $p\myarrowtau p$ holds only when $t$ is the empty trace $\varepsilon$. A process $q$ is \emph{reachable} from $p$ if $p\myarrowtau q$, for some $t\in \act^*$.  
We define the \emph{size}  of an LTS $\mathS=(\proc,\act,\longrightarrow)$, denoted  $|\mathS|$, to be $|\proc|+|{\longrightarrow}|$. The \emph{size of a process} $p\in \proc$, denoted  $|p|$, is the cardinality of $\reach(p)=\{q~|~ q \text{ is reachable from } p\}$ plus the cardinality of the set $\longrightarrow$ restricted to $\reach(p)$. 
A sequence of actions $t\in \act^*$ is a trace of $p$ if there is a $q$ such that $p\myarrowtau q$.
The \emph{depth} of a finite, loop-free process $p$, denoted  $\depth(p)$, is the length of a longest trace $t$ of $p$.
%
In what follows, we shall often describe finite, loop-free processes using the fragment of Milner's CCS~\cite{Milner89} given by the grammar
$p ::= \mathtt{0} ~\mid~ a.p  ~\mid~  p+p$,
where $a \in \act$. For each action $a$ and terms $p,p'$, we write $p \myarrowa p'$ iff
(i) $p =a.p'$ or
(ii) $p =p_1+p_2$, for some $p_1,p_2$, and $p_1\myarrowa p'$ or $p_2\myarrowa p'$ holds.

We consider $n$-nested simulation for $n\geq 1$, and bisimilarity, which are defined below. 

\begin{definition}[\cite{Milner89,GV92}]\label{Def:beh-preorders}
We define each of the following preorders as the largest binary relation over $\proc$ that satisfies the corresponding condition.
\begin{enumerate}[(a)]
    \item \emph{Simulation preorder (S):} $p\curle_{S} q$ iff for all $p\myarrowa p'$ there exists some $q\myarrowa q'$ such that $p'\curle_S q'$.
    \item \emph{$n$-Nested simulation ($n$S)}, where $n\geq 1$, is defined inductively as follows: The $1$-nested simulation preorder $\curle_{1S}$ is $\curle_S$, and the $n$-nested simulation preorder $\curle_{nS}$ for $n > 1$ is the largest relation such that $p\curle_{nS} q$ iff
      (i) for all $p\myarrowa p'$ there exists some $q\myarrowa q'$ such that $p'\curle_{nS} q'$, and
      (ii) $q\curle_{(n-1)S} p$. 
   \item  \emph{Bisimilarity (BS):} $\curle_{BS}$ is the largest symmetric relation satisfying the condition defining $\curle_{S}$. 
\end{enumerate}
\end{definition}

It is well known that bisimilarity is an equivalence relation and all the relations $\curle_{nS}$ 
are preorders~\cite{Milner89,GV92}. We sometimes write $p\sim q$ instead of $p\curle_{BS}q$. Moreover, we have that 
${\sim}\subsetneq{\curle_{nS}}$  and ${\curle_{(n+1)S}}\subsetneq {\curle_{nS}}$ for every $n\geq 1$---see~\cite{Glabbeek01}. We say that $p$ is $n$-nested-simulated by $q$ when $p\curle_{nS} q$.

\begin{definition}[Kernels of the preorders]
    For each $n\geq 1$,  the kernel $\equiv_{nS}$ of $\curle_{nS}$ is the equivalence relation defined thus: for every $p,q\in \proc$, $p\equiv_{nS} q$ iff $p\curle_{nS} q$ and $q\curle_{nS} p$. We say that $p$ and $q$ are $n$-nested-simulation equivalent if $p\equiv_{nS} q$.
\end{definition}


Each relation $\curle_{nS}$, where $n\geq 1$, is characterized (see \Cref{prop:hmt} below) by a fragment $\mathL_{nS}$ of Hennessy-Milner logic, \hml, defined as follows.

\begin{definition}\label{def:mathlx}
For $X\in\{BS\} \cup \{nS\mid n\geq 1\}$, 
$\mathL_X$ is defined by the corresponding grammar given below ($a\in \act$):
\begin{enumerate}[(a)] 
\item $\mathL_S$ ($\mathL_{1S}$): 
$\varphi_S::= ~ \true ~ \mid ~ \ff ~ \mid ~ \varphi_S\wedge \varphi_S~ \mid ~\varphi_S\vee \varphi_S~ \mid ~ \langle a \rangle\varphi_S.$
\item $\mathL_{nS}$, $n\geq 2$: 
$\varphi_{nS}::= ~ \true ~ \mid ~ \ff ~ \mid ~ \varphi_{nS}\wedge \varphi_{nS} ~ \mid ~\varphi_{nS}\vee \varphi_{nS}
~ \mid ~ \langle a \rangle\varphi_{nS} ~ \mid ~ \neg\varphi_{(n-1)S}.$
\item $\hml$ ($\mathL_{BS}$):
    $\varphi_{BS}::= ~ \true ~ \mid ~ \ff ~ \mid ~ \varphi_{BS}\wedge \varphi_{BS} ~ \mid ~\varphi_{BS}\vee \varphi_{BS}~ \mid ~ \langle a \rangle\varphi_{BS} ~ \mid ~ [a]\varphi_{BS} ~ \mid ~ \neg\varphi_{BS}.$
\end{enumerate}
\end{definition}

Note that the explicit use of negation in the grammar for $\mathL_{BS}$ is unnecessary. However, we included the negation operator explicitly so that $\mathL_{BS}$ extends syntactically each of the other modal logics presented in Definition~\ref{def:mathlx}. 

Given a formula $\varphi\in\mathL_{BS}$, the \emph{modal depth} of $\varphi$, denoted  $\md(\varphi)$, is the maximum nesting of modal operators in $\varphi$. We define the \emph{size} of formula $\varphi$, denoted  $|\varphi|$, to be the number of symbols in $\varphi$. Finally, $\sub(\varphi)$ denotes the set of subformulae of formula $\varphi$.

Truth in an LTS $\mathS=(\proc,\act,\longrightarrow)$ is defined via the satisfaction relation $\models$ as follows, where we omit the standard clauses for the Boolean operators:
$$\begin{aligned}
    &p\models \langle a\rangle\varphi \text{ iff there is 
 some } p\myarrowa q \text{ such that } q\models\varphi;\\
    &p\models [a]\varphi \text{ iff for all } p\myarrowa q \text{ it holds that } q\models\varphi  .
\end{aligned}$$
If $p\models \varphi$, we say that $\varphi$ is true, or satisfied, in $p$. A formula $\varphi$ is \emph{satisfiable} if there is a process that satisfies it. 
 Formula $\varphi_1$ \emph{entails} $\varphi_2$, denoted  $\varphi_1\models \varphi_2$, if every process
that satisfies  $\varphi_1$ also satisfies   $\varphi_2$. Moreover, $\varphi_1$ and $\varphi_2$ are \emph{logically equivalent}, denoted  $\varphi_1\equiv\varphi_2$, if $\varphi_1\models \varphi_2$ and $\varphi_2\models \varphi_1$. For example, 
$\varphi_1\wedge(\varphi_2\vee\varphi_3)\equiv (\varphi_1\wedge\varphi_2)\vee(\varphi_1\wedge\varphi_3)$, 
for every $\varphi_1,\varphi_2,\varphi_3\in\mathL_{BS}$.

Given a process $p$ and $\mathL \subseteq \mathL_{BS}$, we define  $\mathL(p)=\{\varphi\in\mathL \mid p\models\varphi\}$. A simplification of the Hennessy-Milner theorem gives a modal characterization of bisimilarity over finite processes. An analogous result is true for every preorder examined in this paper.

\begin{proposition}\label{prop:hmt}[\cite{HennessyM85,GV92}]\label{logical_characterizations}
     For all processes $p,q$ in a finite LTS, $p\sim q$ iff $\mathL_{BS}(p)=\mathL_{BS}(q)$. Moreover, $p\curle_{nS} q$ iff $\mathL_{nS}(p)\subseteq \mathL_{nS}(q)$ for each $n\geq 1$.
\end{proposition}

\begin{definition}[\cite{BoudolL92,AFEIP11}]\label{def:prime-formula}
Let $\mathL\subseteq \mathL_{BS}$. 
A formula $\varphi\in \mathL_{BS}$ is \emph{prime in $\mathL$} if  $\varphi\models\varphi_1 \vee \varphi_2$ implies $\varphi\models \varphi_1$ or  $\varphi\models \varphi_2$, for all $\varphi_1,\varphi_2\in \mathL$. 
\end{definition}

When the logic $\mathL$ is clear from the context, we say that $\varphi$ is prime. Note that every unsatisfiable formula is trivially prime in $\mathL$, for every $\mathL$.

\begin{example}\label{ex:prime}
    The formula $\langle a\rangle \true$ is prime in $\mathL_S$. Indeed, let $\varphi_1,\varphi_2\in\mathL_S$ and assume that $\langle a\rangle \true\models\varphi_1\vee\varphi_2$. Since $a.\mathtt{0} \models\langle a\rangle \true$, without loss of generality, we have that $a.\mathtt{0} \models\varphi_1$. We claim that $\langle a\rangle\true\models\varphi_1$. To see this, let $p$ be some process such that $p\models\langle a\rangle\true$---that is, a process such that $p \myarrowa p'$ for some $p'$. It is easy to see that $a.\mathtt{0}\curle_S p$. Since $a.\mathtt{0} \models \varphi_1$, Proposition~\ref{logical_characterizations} yields that $p\models\varphi_1$, proving our claim and the primality of $\langle a\rangle\true$. On the other hand, the formula $\langle a\rangle\true\vee\langle b\rangle\true$ is not prime in $\mathL_S$. Indeed, $\langle a\rangle\true\vee\langle b\rangle\true\models \langle a\rangle\true\vee\langle b\rangle\true$, but neither $\langle a\rangle\true\vee\langle b\rangle\true \models\langle a\rangle\true$  nor $\langle a\rangle\true\vee\langle b\rangle\true \models\langle b\rangle\true$ hold. 
\end{example}

Characteristic formulae are defined next, with two distinct definitions: within logic \mathL, and another modulo an equivalence relation.

\begin{definition}[\cite{AILS07,GrafS86a,SteffenI94}]\label{def:characteristic}
Let $\mathL\subseteq \mathL_{BS}$. 
A formula  $\varphi\in\mathL$ is \emph{characteristic  for $p\in \proc$ within $\mathL$} iff, for all $q \in \proc$, it holds that $q \models \varphi\Leftrightarrow\mathL(p) \subseteq \mathL(q)$. 
\end{definition}

\begin{proposition}[\cite{AcetoMFI19}]\label{prop:charact-via-primality}
For every $n\geq 1$, $\varphi\in \mathL_{nS}$ is characteristic for some process within $\mathL_{nS}$ iff $\varphi$ is satisfiable and prime in $\mathL_{nS}$.
\end{proposition}

\begin{remark}
We note, in passing, that the article~\cite{AcetoMFI19} does not deal explicitly with $nS$, $n\geq 3$. However, its results apply to all the $n$-nested simulation preorders.
\end{remark}


\begin{definition}\label{def:characteristic-equivalence}
Let $X\in\{BS\} \cup \{nS\mid n\geq 1\}$. 
A formula $\varphi\in \mathL_X$ is characteristic for $p\in \proc$ modulo $\equiv_X$  iff for all $q\in \proc$, it holds that $q\models \varphi\Leftrightarrow \mathL_X(p)=\mathL_X(q)$.
\end{definition} 

\begin{proposition}\label{prop:char-mod-equiv}
   Let $X\in\{BS\} \cup \{nS\mid n\geq 1\}$. 
   A formula $\varphi\in \mathL_X$ is characteristic for a process modulo $\equiv_X$ iff $\varphi$ is satisfiable and for every $p,q\in\proc$ such that $p\models\varphi$ and $q\models\varphi$, $p\equiv_X q$ holds.
\end{proposition}

\paragraph{The \hml tableau}
Let $S$ be a set of formulae. We write
$\bigwedge S$ for $\bigwedge_{\varphi\in S}\varphi$, when $S$ is finite, and
 $\sub(S)$ for $\{\varphi ~ | ~\varphi\in\sub(\psi) \text{ for some } \psi\in S\}$. Note that $\sub(S)$ is finite, when so is $S$.

\begin{definition}\label{def:proptableau}
Let $T$ be a set of \hml formulae. 
\begin{enumerate}[(a)]
    \item $T$ is \emph{propositionally inconsistent} if 
    $\ff\in T$, or $\psi\in T$ and $\neg\psi\in T$ for some formula $\psi$. Otherwise, $T$ is \emph{propositionally consistent}.
    \item  $T$ is a \emph{propositional tableau} if the following conditions are met:
    \begin{enumerate} [(i)]
        \item if $\psi\wedge\psi'\in T$, then $\psi,\psi'\in T$,
        \item if $\psi\vee\psi'\in T$, then either $\psi\in T$ or $\psi'\in T$, and
        \item $T$ is propositionally consistent. 
    \end{enumerate}
\end{enumerate}
\end{definition}

\begin{definition}\label{def:hmltableau}
Let $\act=\{a_1,\dots,a_k\}$. An \hml tableau is a tuple $T=(S,L,R_{a_1},\dots,R_{a_k})$, where $S$ is a set of states, $L$ is a labelling function that maps every $s\in S$ to a set $L(s)$ of formulae, and  $R_{a_i}\subseteq S\times S$, for every $1\leq i\leq k$, such that
\begin{enumerate}[(i)]
    \item $L(s)$ is a propositional tableau for every $s\in S$,
    \item if $[a_i] \psi\in L(s)$ and $(s,t)\in R_{a_i}$, then $\psi\in L(t)$, and
    \item if $\langle a_i\rangle\psi\in L(s)$, then there is some $t$ such that $(s,t)\in R_{a_i}$ and $\psi\in L(t)$.
\end{enumerate}
\end{definition}

\noindent An \hml tableau for $\varphi$ is an \hml tableau such that $\varphi\in L(s)$ for some $s\in S$.


\begin{proposition}[\cite{HalpernM92}]\label{tableausat}
An \hml formula $\varphi$ is satisfiable iff there is an \hml tableau for $\varphi$.
\end{proposition}

\begin{remark}\label{rem:tableausat}
    The proof of the ``right-to-left'' direction of Proposition~\ref{tableausat} constructs an LTS corresponding to a process satisfying $\varphi$ from an \hml tableau for $\varphi$ in a straightforward way: given an \hml tableau $T=(S,L,R_{a_1},\dots,R_{a_k})$ for $\varphi$, define the LTS $\mathS=(P,\myarrowasubone,\dots,\myarrowasubk)$, where $P=S$ and every $\myarrowasubi$ coincides with $R_{a_i}$. Note that $T$ and $\mathS$ have the same size and depth.
\end{remark}
%
%
%
%
\paragraph{Complexity and games}
In what follows, we 
shall use the following results from~\cite{AcetoACI25}.

\begin{proposition}[\cite{AcetoACI25}]\label{prop:sat}
 Let $|\act|>1$.
  \begin{enumerate}[(a)]
      \item Satisfiability of formulae in $\mathL_{S}$ is in \cP.
      \item Satisfiability of formulae in $\mathL_{2S}$ is \NP-complete.
      \item Satisfiability of formulae in $\mathL_{3S}$ is \pspace-complete.
  \end{enumerate}  
\end{proposition}

We refer to the problem of determining whether a formula $\varphi\in\mathL$ is prime in $\mathL$ as the  \emph{Formula Primality Problem for $\mathL$}.
Hardness results for this problem follow.

\begin{proposition}[\cite{AcetoACI25}]\label{prop:primality}
    Let $|\act|>1$.
    \begin{enumerate}[(a)]
    \item The Formula Primality Problem for $\mathL_{2S}$ is \conp-hard.
    \item The Formula Primality Problem for $\mathL_{3S}$ is \pspace-hard.
    \end{enumerate}
\end{proposition}

The following corollary follows from the results of~\cite{AcetoACI25} presented above.

\begin{corollary}\label{cor:n-sat}
     Let $|\act|>1$ and $n\geq 3$. Satisfiability of formulae in $\mathL_{nS}$ is \pspace-complete, and the Formula Primality Problem for $\mathL_{nS}$ is \pspace-hard.
\end{corollary}

We introduce 
two complexity classes that play an important role in pinpointing the complexity of deciding characteristic formulae within $\mathL_{2S}$. The first class is $\dpc=\{L_1\cap L_2 \mid L_1\in\NP \text{ and } L_2\in \conp\}$~\cite{PapadimitriouY84} and the second one is  \us~\cite{BlassG82}, which 
is defined thus:
A language $L\in\us$ iff there is a non-deterministic Turing machine $T$ such that, for every instance $x$ of $L$, $x\in L$ iff $T$ has \emph{exactly one} accepting path on input $x$. The problem $\textsc{UniqueSat}$, viz.~the problem of deciding whether a given Boolean formula has exactly one satisfying truth assignment, is \us-complete. Note that $\us\subseteq\dpc$~\cite{BlassG82}.

\begin{proposition}[\cite{AcetoACI25}]\label{prop:char-ns-complexity}
\begin{enumerate}[(a)]
    \item Let $|\act|>1$ and $\varphi\in\mathL_{2S}$. Deciding whether $\varphi$ is characteristic for a process within $\mathL_{2S}$ (respectively, modulo $\equiv_{2S}$) is \us-hard (respectively, \conp-hard).
    \item Let $|\act|>1$ and $\varphi\in\mathL_{nS}$, where $n\geq 3$. Deciding whether $\varphi$ is characteristic for a process within $\mathL_{nS}$ (or modulo $\equiv_{nS}$) is \pspace-hard.
\end{enumerate}
\end{proposition}

An \emph{alternating Turing machine} is a non-deterministic Turing machine whose set of states is partitioned into existential and universal states. An existential state is accepting if at least one of its transitions leads to an accepting state. In contrast, a universal state is accepting only if all its transitions lead to accepting states. The machine as a whole accepts an input if its initial state is accepting. The complexity class \ap is the class of languages accepted by polynomial-time alternating Turing machines. 
An \emph{oracle Turing machine} is a Turing machine that has access to an oracle---a ``black box'' capable of solving a specific computational problem in a single operation. An oracle Turing machine can perform all the usual operations of a Turing machine, and can also query the oracle to obtain a solution to any instance of the computational problem for that oracle. We use $\textsf{C}^{\textsf{C'}[\text{poly}]}$ to denote the complexity class of languages decidable by an algorithm in class $\textsf{C}$ that makes polynomially many oracle calls to a language in $\textsf{C'}$. Note, for example, that $\pspace^{\pspace[\text{poly}]}=\pspace$, since a polynomial-space oracle Turing machine can simulate any \pspace oracle query by solving the problem itself in polynomial space.

\begin{proposition}
\label{prop:pspace-ap}
\begin{enumerate}[(a)]
    \item  \textnormal{(\cite{DBLP:journals/jacm/ChandraKS81})\textbf{.}} $\ap=\pspace$.
    \item $\ap^{\pspace[\text{poly}]}=\pspace$.
\end{enumerate}
\end{proposition}



Consider now two-player games that have the following characteristics: they are \emph{zero-sum} (that is, player one's gain is equivalent to player two's loss), \emph{perfect information} (meaning that, at every point in the game, each player is fully aware of all events that have previously occurred), \emph{polynomial-depth} (i.e.~the games proceed for a number of rounds that is polynomial in the input size), and \emph{computationally bounded} (that is, for each round, the computation performed by a player can be simulated by a Turing machine in polynomial time in the input size). For two-player games that have all four characteristics described above, there is a polynomial-time alternating Turing machine that decides whether one of the players has a winning strategy~\cite{Feigenbaum1998}. The two-player games we will introduce in the following section are zero-sum, perfect-information, and polynomial-depth, but they are not computationally bounded: In each round, at most a polynomial number of problems in \pspace must be solved. We call these games zero-sum, perfect-information, polynomial-depth \emph{with a \pspace oracle}. Then, the polynomial-time alternating Turing machine that determines whether one of the players has a winning strategy for such a game has to use a polynomial number of oracle calls to \pspace problems in order to correctly simulate the game. Thus, using Proposition~\ref{prop:pspace-ap}(b), we have that: 
\begin{corollary}\label{cor:pspace-games}
   For a two-player, zero-sum, perfect-information, polynomial-depth game with a \pspace oracle, we can decide whether a player has a winning strategy in polynomial space.
\end{corollary}

\section{The complexity of deciding characteristic formulae within $\mathL_{nS}$, $n\geq 3$}\label{section:deciding-char-nS}

Since the characteristic formulae within $\mathL_{nS}$ coincide with the satisfiable and prime ones, and satisfiability in $\mathL_{nS}$  is \pspace-complete for $n\geq 3$,
we investigate the complexity of the Formula Primality Problem for $\mathL_{nS}$, where $n\geq 3$. In this section, we present a polynomial-space algorithm that solves this problem, matching the lower bound from Proposition~\ref{prop:primality}. To this end, we introduce two families of games: the \emph{char-for-n-nested-simulation-equivalence} game, referred to as the \nsimeq game, for every $n\geq 1$, and the \emph{prime-for-n-nested-simulation-preorder} game, referred to as the \nsimpre game, for every $n\geq 3$.

Assume that $\varphi$ is a satisfiable formula in $\mathL_{nS}$. All of the games are played between players $A$ and~$B$. If a game is initiated on $\varphi$, it starts with two or three states, each of which has a label equal to~$\{\varphi\}$.  As the game proceeds, the players extend the already existing structures and explore (two or three) tableaux for $\varphi$ that satisfy some additional, game-specific conditions. Player $A$ has a winning strategy for the \nsimeq game iff every two processes that satisfy $\varphi$ are equivalent modulo $\equiv_{nS}$---that is, $\varphi$ is a characteristic formula modulo $\equiv_{nS}$. The existence of a winning strategy for player $A$ in the \nsimpre game on $\varphi$ is equivalent to the primality of $\varphi$ in $\mathL_{nS}$.  A difference between the games \nsimeq and \nsimpre is that the former can be initiated on a satisfiable formula that belongs to $\mathL_{\ell S}$, where $\ell\geq n$, whereas the latter is only started on a satisfiable formula that is in $\mathL_{nS}$. When $A$ and $B$ play one of the games \nsimeq or \nsimpre, at some point, they have to play the \nmosimeq game initiated on states labelled with possibly different finite subsets of $\mathL_{nS}$ formulae. This is why the \nsimeq game is generalized to start with such labelled states. 

For the presentation of the games, let $\act=\{a_1,\dots,a_k\}$. Basic moves that $A$ and $B$ can play are presented in Table~\ref{tab:moves}. 

\begin{table}
\begin{center}
\begin{tabular}{ | P{1.9cm} | P{13cm} | } 
\hline
Move name  & Move description\\  
 \hline
Pl($\wedge$) & For every $\psi_1\wedge\psi_2\in L_i(p)$,  $Pl$ replaces $\psi_1\wedge\psi_2$ with both $\psi_1$ and $\psi_2$ in $L_i(p)$. \\  
 \hline
Pl($\vee$) & For every $\psi_1\vee\psi_2\in L_i(p)$,  $Pl$  chooses $\psi\in\{\psi_1,\psi_2\}$ and replaces $\psi_1\vee\psi_2$ with $\psi$ in $L_i(p)$.\\  
 \hline
Pl($\Diamond$) & For every $\langle a_j\rangle \psi\in L_i(p)$,  $Pl$ adds a new state $p'$ to $S_i$, $(p,p')$ to $R_{a_j}^i$, and sets $L_i(p')=\{\psi\}\cup\{\psi'\mid [a_j]\psi'\in L_i(p)\}$.\\  
 \hline
B($\square$) & $B$ chooses between doing nothing and picking some $1\leq j\leq k$. In the latter case, $B$ adds a new state $p'$ to $S_i$, $(p,p')$ to $R^i_{a_j}$, and  sets $L_i(p')=\{\psi ~ | ~ [a_j]\psi\in L_i(p)\}$.  \\  
 \hline
 A(sub) & For every $\psi\in\sub(\varphi)$, $A$ chooses between adding or not adding $\psi$ to $L_i(p)$.  \\  
 \hline
 A(rem) & For every $j$-successor $p'$ of $p$, $A$ removes $p'$ from $T_i$ if there is a $j$-successor $p''$ of $p$, such that $p'\neq p''$ and $L_i(p')\subseteq L_i(p'')$. \\  
 \hline
\end{tabular}
\end{center}
\caption{Basic moves that players $A$ and $B$ can play in any game initiated on formula $\varphi$. The description is for player $Pl\in\{\text{A,B}\}$ who plays on state $p\in S_i$, where $i\in\{1,2,3\}$, and action $a_j\in\act$.
}
\label{tab:moves}
\end{table}

\subsection{The \nsimeq game, $n\geq 1$}\label{subsec:the-simulation-game}

We present the first family of games.  We begin by describing the \simequiv game, followed by the \nsimeq game for $n\geq 2$. Let $\varphi$ be a satisfiable formula in $\mathL_{\ell S}$, where $\ell\geq n$. The games are defined so  
that player $A$ has a winning strategy for the \nsimeq game played on $\varphi$ iff every two processes satisfying $\varphi$ are $n$-nested-simulation equivalent: we prove this statement for the \simequiv game, and assuming that this is true for the \nmosimeq game, we show the statement for the \nsimeq game.

\paragraph{The \simequiv game}
We first introduce the \simequiv game started on $\varphi$. During the game, $B$ constructs two labelled trees $T_1$ and $T_2$ that correspond to two arbitrary processes $p_1$ and $p_2$ satisfying $\varphi$ and challenges $A$ to construct a simulation relation between the states of $T_1$ and $T_2$ showing that $p_1\curle_S p_2$. The labelled trees constructed by $B$ are denoted $T_1=(S_1,L_1,R^1_{a_1},\dots,R^1_{a_k})$ and $T_2=(S_2,L_2,R^2_{a_1},\dots,R^2_{a_k})$, and the game starts with $S_1=\{p_0^1\}$ and $S_2=\{p_0^2\}$, $L_1(p_0^1)=L_2(p_0^2)=\{\varphi\}$, and $R^i_{a_j}=\emptyset$, for every $i=1,2$ and $1\leq j\leq k$. We describe the $l$-th round of the game, where $l\geq 1$, in Table~\ref{tab:simulation-game}. States $p_1,p_2$ are $p_0^1,p_0^2$ respectively, if $l\in\{1,2\}$, or the two states that $B$ and $A$ respectively chose at the end of round~$l-1$, if~$l>2$. For two states $p,p'$ such that $(p,p')\in R_{a_j}$, we say that $p'$ is a $j$-successor of $p$. We use $p$, $p_1$, $p_2$, etc.~to denote both processes and states of the labelled trees; the intended meaning will be clear from the context.
\begin{table}
\begin{center}
\begin{tabular}{ | p{15cm} | } 
\hline
\textbf{$\mathbf{1}^{\text{st}}$ round.} 
 $B$ plays moves $B(\wedge)$ and $B(\vee)$ on $p_i$, for both $i=1,2$, until no formula can be replaced in $L_i(p_i)$. 
 If $\bigwedge L_i(p_i)$ becomes unsatisfiable, then $B$ loses.

\textbf{$\mathbf{l}^{\text{th}}$ round, $\mathbf{l\geq 2}$.} 
\begin{enumerate}
    \item For every $a_j\in \act$, $B$ plays as follows. He plays move $B(\Diamond)$ on $p_i$ for both $i=1,2$, and move $B(\Box)$ only on $p_1$. Then, for both $i=1,2$, $B$ plays moves $B(\wedge)$ and $B(\vee)$  on every $p_i'$ such that $(p_i,p_i')\in R_{a_j}^i$ until no formula can be replaced in $L_i(p_i')$. If $\bigwedge L_i(s)$ becomes unsatisfiable for some $i=1,2$ and $s\in S_i$, then $B$ loses. 
    \item $B$ chooses a $1\leq j\leq k$ and a $j$-successor $p_1'$ of $p_1$. If $p_1$ has no $j$-successors, then $B$ loses.
    \item $A$ chooses a $j$-successor $p_2'$  of $p_2$. If $p_2$ has no $j$-successors, then $A$ loses.
    \item The $l+1$-th round starts on $p_1'$, $p_2'$.
\end{enumerate}  \\
 \hline
\end{tabular}
\end{center}
\caption{The \simequiv game initiated on a satisfiable $\varphi\in\mathL_{\ell S}$, where $\ell\geq 1$.}
\label{tab:simulation-game}
\end{table}

\begin{example}\label{ex:simulation-game}
   (a) Let $\mathbf{0}$ denote $\bigwedge_{i=1}^k [a_k]\ff$ and consider the formula $\varphi=\langle a_1\rangle \mathbf{0}$. Note that both the processes $r_1=a_1.\mathtt{0}$ and $r_2=a_1.\mathtt{0}+a_2.\mathtt{0}$ satisfy $\varphi$, and $r_1\not\equiv_S  r_2$. Therefore, player $B$ should have a winning strategy for the \simequiv game on $\varphi$, which is true as $B$ can play as follows. At the first round, he can make no replacement in $L_i(p_i)$ for both $i=1,2$. At step~1 of the second round, he generates states $p_1'$ and $p_2'$ that are $1$-successors of $p_1$ and $p_2$ respectively, when he plays move B($\Diamond$). Then, when $B$ plays move B($\square$) on $p_1$, he chooses to generate state $p_1''$ that is a $2$-successor of $p_1$, adds $(p_1,p_1'')$ to $R^1_{a_2}$, and  sets $L_1(p_1'')=\emptyset$. He applies move B($\wedge$) on $p_i'$  to obtain $L_i(p_i')=\{[a_1]\ff,\dots,[a_k]\ff\}$ for $i=1,2$. At step~2, player~$B$ chooses $p_1''$ and since $p_2$ has no $2$-successors, $A$ loses at step~3.

    (b) On the other hand, player $A$ has a winning strategy for the \simequiv game initiated on $\psi=\langle a_1\rangle \mathbf{0}\wedge \bigwedge_{i=2}^k [a_i]\ff$. (Note that the process $r_1=a_1.\mathtt{0}$ is the unique process modulo $\equiv_S$ that satisfies $\psi$.) After completing the first round, $B$ generates two states $p_1'$ and $p_2'$ which are $1$-successors of $p_1$ and $p_2$ respectively, and sets $L_1(p_1')=L_2(p_2')=\{\mathbf{0}\}$ when applying move B($\Diamond$). If he chooses to generate a $j$-successor $p_1''$ of $p_1$, where $j\neq 1$, when he plays move B($\square$), then he loses, since $L_1(p_1'')=\{\ff\}$ is unsatisfiable. So, he chooses to do nothing at move B($\square$) and picks $p_1'$ at step~2. Then, $A$ picks $p_2'$ at step 3. In round~3, player $B$ either generates a $j$-successor of $p_1'$ for some $1\leq j\leq k$ when applying move B($\square$) and loses because the label set of the new state is unsatisfiable or generates no successors and loses at step~2. 
    \end{example}

The labelled trees $T_1, T_2$, constructed during the \simequiv game on $\varphi$, form partial tableaux for $\varphi$. This is because some states are abandoned during the game, which may result in $T_1, T_2$ failing to satisfy condition (iii) of Definition~\ref{def:hmltableau}. The \simequiv game can be generalized so that it starts with  $S_1=\{s_1\}$, $S_2=\{s_2\}$, $R_{a_j}^i$ being empty for every $i=1,2$ and $1\leq j\leq k$, and $L_1(s_1)=U_1$, $L_2(s_2)=U_2$, where $U_1,U_2$ are finite subsets of $\mathL_{\ell S}$, $\ell\geq 1$. We denote by $\simab(U_1,U_2)$ the \simequiv game that starts from the configuration just described. In particular, $\simab(\{\varphi\},\{\varphi\})$ is called the \simequiv game on $\varphi$.

 Let $p\in S_i$, where $i=1,2$. We denote by $\inlabeli(p)$ the initial label of $p$ before moves B($\wedge$) and B($\vee$) are applied on $p$ and $\finlabeli(p)$ the final label of $p$ after moves B($\wedge$) and B($\vee$) have been applied on $p$ (until no formula can be replaced in $L_i(p)$).
As shown in Example~\ref{ex:simulation-game}, in the \simequiv game on $\varphi$, player $B$ consistently plays $T_i$, $i=1,2$, on a process $r$  satisfying $\varphi$. Intuitively, $T_i$ represents part or all of $r$, viewing states in $S_i$ as processes reachable from $r$ and $R^i_{a_j}$ as transitions. The formal definition follows.
\begin{definition}\label{def:B-plays-consistently}
    Assume that the \simequiv game is played on $\varphi\in\mathL_{\ell S}$, $\ell\geq 1$. We say that $B$ plays $T_i$, where $i\in\{1,2\}$, consistently on a process $r$ if there is a mapping $map:S_i\rightarrow \proc$ such that the following conditions are satisfied:
  \begin{enumerate}
      \item for every $p\in S_i$, $map(p)\models \bigwedge \finlabeli(p)$,
      \item for every $(p,p')\in R^i_{a_j}$, $map(p)\myarrowasubj map(p')$, and 
      \item $map(p_0^i)=r$, where $p_0^i$ is the initial state of $T_i$.
  \end{enumerate}
\end{definition}

    Next, we prove that player $A$ has a winning strategy for the \simequiv game on $\varphi$ iff every two processes that satisfy $\varphi$ are simulation equivalent.

    \begin{restatable}{proposition}{simgamecorrect}\label{prop:simul-game}
     Let $\varphi\in\mathL_{\ell S}$, where $\ell\geq 1$, be a satisfiable formula. Player $A$ has a winning strategy for the \simequiv game on $\varphi$ iff  $r_1\equiv_S r_2$, for every two processes $r_1,r_2$ that satisfy $\varphi$.
    \end{restatable}
    \begin{proofsketch}
        Assume that every two processes that satisfy $\varphi$ are simulation equivalent. If $B$ plays in the \simequiv game on $\varphi$ such that $L_i(s)$ always remains satisfiable for every $i=1,2$ and $s\in S_i$, then he plays the labelled tree $T_i$, $i=1,2$, consistently on a process $r_i$ that satisfies $\varphi$. That is, there exists a mapping $map:S_1\cup S_2\rightarrow \proc$ satisfying conditions 1--3 of Definition~\ref{def:B-plays-consistently}. Since $r_1\equiv_S r_2$, for any $j$-successor $p_1'$ chosen by $B$ at step 2 of a round, player $A$ can respond with a $j$-successor $p_2'$ at step 3, such that $map(p_1') \curle_S map(p_2')$. Thus, $A$ has a strategy to avoid losing in any round. Moreover, within at most $\md(\varphi) + 1$ rounds, $B$ either produces an unsatisfiable $L_i(s)$ or fails to generate any $j$-successors for $1 \leq j \leq k$, and hence loses. To prove the converse, we proceed by showing the contrapositive. Let $q_1,q_2$ be two processes that satisfy $\varphi$ and $q_1\not\curle_S q_2$. We can assume w.l.o.g.\ that $\depth(q_i)\leq\md(\varphi)+1$ for both $i=1,2$. Then, $B$ can play $T_i$ consistently on $q_i$ for both $i=1,2$ and, while constructing $T_1$, include a trace from $q_1$ that witnesses the failure of $q_1\curle_S q_2$. By the definition of $\curle_S$, there is $a_i\in\act$ such that either
         (a) $q_1\myarrowasubi$ and $q_2\notmyarrowasubi$, or
         (b) $q_1\myarrowasubi q_1'$ for some $q_1'$ such that $q_1'\not\curle_S q_2'$, for every $q_2\myarrowasubi q_2'$. 
        In  case (a), $B$ plays move B($\square$) in the second round to generate an $i$-successor $p_1'$ of $p_0^1$ and picks $p_1'$ at step 2. Since $T_2$ is played consistently on $q_2$ and $q_2\notmyarrowasubi$, $p_0^2$ has no $i$-successors after step 1 of round 2. Thus, $A$ cannot respond and loses at step 3. In  case (b), $B$ plays similarly: he plays B($\square$) to generate and pick an $i$-successor $p_1'$ and $A$ responds by picking an $i$-successor $p_2'$ of $p_0^2$. It holds that $map(p_1')\not\curle_S map(p_2')$ and $B$ can recursively apply the same strategy on $p_1'$ and $p_2'$. Since $\depth(q_i)\leq\md(\varphi)+1$, within at most $\md(\varphi)+2$ rounds, $B$ will generate some $j$-successor using B($\square$) for which $A$ has no response and, consequently, $A$ loses.
        Based on the above, the game always terminates no later than the ($\md(\varphi)+2$)-th round. By induction, using standard arguments for two-player, zero-sum games, we can show that either $A$ has a winning strategy or $B$ has a winning strategy. Then, the proposition holds.
    \end{proofsketch}

    \begin{proposition}\label{prop:sim-win-algo}
    Let $\varphi\in\mathL_{\ell S}$, $\ell\geq 1$, be a satisfiable formula. Deciding whether every two processes $p_1,p_2$ that satisfy $\varphi$ are equivalent modulo $\equiv_S$ can be done in polynomial space.
    \end{proposition}
     \begin{proof}
     From Proposition~\ref{prop:simul-game}, it suffices to show that determining whether player $A$ has a winning strategy for the \simequiv game on $\varphi$ can be done in polynomial space. The \simequiv game is a zero-sum and perfect-information game. It is also of polynomial depth, since it stops after at most $\mod(\varphi)+2$ rounds. Finally, in every round, the satisfiability of $\bigwedge L_i(s)$ has to be checked a polynomial number of times. If $\varphi\in\mathL_S$, then Proposition~\ref{prop:sat}(a) yields that the game is computationally bounded. If $\varphi\in\mathL_{\ell S}$,  for some $\ell\geq 2$, then the \simequiv game is a game with a \pspace oracle by  Proposition~\ref{prop:sat} and Corollary~\ref{cor:n-sat}. The desired conclusion then follows from Corollary~\ref{cor:pspace-games}.
     \end{proof}

     By determining whether $A$ has a winning strategy for the game $\simab(U_1,U_2)$, where $U_1,U_2$ are two finite subsets of formulae, we can decide whether every process that satisfies $\bigwedge U_1$ is simulated by any process that satisfies $\bigwedge U_2$. Therefore, the latter problem also lies in \pspace. 

     \begin{corollary}\label{cor:simul-game}
   Let $U_1,U_2$ be finite sets of $\mathL_{\ell S}$, $\ell\geq 1$. Player $A$ has a winning strategy for $\simab(U_1,U_2)$ iff $p_1\curle_S p_2$, for every $p_1,p_2$ such that $p_1\models\bigwedge U_1$ and $p_2\models\bigwedge U_2$.
    \end{corollary}

     \begin{restatable}{corollary}{corsimgame}\label{cor:simul-game-complexity}
    Let $U_1,U_2$ be finite sets of $\mathL_{\ell S}$, $\ell\geq 1$. Deciding whether $p_1\curle_S p_2$ is true for every two processes $p_1,p_2$ such that $p_1\models\bigwedge U_1$ and $p_2\models\bigwedge U_2$ can be done in polynomial space.
    \end{restatable}

\paragraph{The \nsimeq game, $n\geq 2$}

Let $n\geq 2$. We denote by $\isimab(U_1,U_2)$, where $i\geq 2$, the \isimeq  that starts with  $S_1=\{s_1\}$, $S_2=\{s_2\}$, $R_{a_j}^t$ being empty for every $t=1,2$ and $1\leq j\leq k$, and $L_1(s_1)=U_1$, $L_2(s_2)=U_2$ with $U_1,U_2$ being finite subsets of $\mathL_{\ell S}$, $\ell\geq i$. Specifically, $\isimab(\{\varphi\},\{\varphi\})$ is called the \isimeq game on $\varphi$. We say that the \isimeq game is correct if Propositions~\ref{prop:simul-game} and~\ref{prop:sim-win-algo} and Corollaries~\ref{cor:simul-game} and~\ref{cor:simul-game-complexity} hold, when $\curle_S$ and $\equiv_S$ are replaced by $\curle_{iS}$ and $\equiv_{iS}$, respectively, $\mathL_{\ell S}$, with $\ell\geq 1$, by $\mathL_{\ell S}$, with $\ell\geq i$, and $\simab$ by $\isimab$. Assume that the \nmosimeq game has been defined so that it is played by players $A$ and $B$ and it is correct.

We can now describe the \nsimeq game on $\varphi$. 
Each round of the \nsimeq game on $\varphi$ follows the steps of the respective round of the \simequiv game on $\varphi$ and includes some additional steps. Analogously to the \simequiv game, if $A$ wins the \nsimeq game on $\varphi$, then the labelled trees $T_1, T_2$, constructed during the game, will correspond to two processes  $p_1,p_2$ such that $p_i\models \varphi$ for both $i=1,2$, and $p_1\curle_{nS} p_2$. By the definition of $\curle_{nS}$, a necessary condition for $p_1\curle_{nS}p_2$ is $p_1\equiv_{(n-1)S} p_2$. This fact is the intuition behind the step preceding the first round and steps 5--6 of the game described below.

The \nsimeq game on $\varphi$ starts with $A$ and $B$ playing the \nmosimeq game on $\varphi$. If $A$ wins, the \nsimeq game resumes. Otherwise, $A$ loses the \nsimeq game on $\varphi$. During the game, $B$ constructs two labelled trees, denoted  $T_1=(S_1,L_1,R^1_{a_1},\dots,R^1_{a_k})$ and $T_2=(S_2,L_2,R^2_{a_1},\dots,R^2_{a_k})$. The first round starts with $S_1=\{p_0^1\}$, $S_2=\{p_0^2\}$, $L_1(p_0^1)=L_2(p_0^2)=\{\varphi\}$, and all $R^i_{a_j}$ being empty, $i=1,2$, and is the same as the first round of the \simequiv game. For $l\geq 2$, the $l$-th round of the game includes steps 1--4 of the $l$-th round of the \simequiv game together with the following steps.
\begin{enumerate}\setcounter{enumi}{4}
    \item $A$ and $B$ play two versions of the \nmosimeq game:  $\nmosimab(L_1(p_1'),L_2(p_2'))$ and $\nmosimab(L_2(p_2'),$ $L_1(p_1'))$.
    \item If $A$ wins both versions of the \nmosimeq game at step 5,  round $l+1$ of the \nsimeq game starts on $p_1',p_2'$. Otherwise, $A$ loses.
\end{enumerate}

From the assumption that the \nmosimeq game is correct and using  arguments analogous to those we employed to prove the correctness of the \simequiv game, the \nsimeq game is also correct.

We already know that no formula in $\mathL_S$ is characteristic modulo $\equiv_S$, and so this problem is trivial~\cite{AcetoACI25}. From Proposition~\ref{prop:char-ns-complexity}, the problem is \pspace-hard for $\mathL_{nS}$, $n\geq 3$, and from this subsection, we derive the following result. We examine the same problem for $\mathL_{2S}$ in Section~\ref{section:deciding-char-2S}.

\begin{restatable}{corollary}{ncharequiv}\label{cor:n-char-equiv}
    Let $|\act|>1$. Deciding whether a formula $\varphi\in\mathL_{nS}$, where $n\geq 3$, is characteristic for a process modulo $\equiv_{nS}$ is \pspace-complete.
\end{restatable}





\subsection{The \nsimpre game, $n\geq 3$}\label{subsec:the-n-prime-game}

Let $n\geq 3$. We use the \nsimpre game  on a satisfiable $\varphi\in\mathL_{nS}$ to check whether $\varphi$ is prime in $\mathL_{nS}$, and thus characteristic for a process within $\mathL_{nS}$.
To this end, the \nsimpre game is developed so that $A$ has a winning strategy iff for every two processes $p_1,p_2$ satisfying $\varphi$ there is a process $q$ satisfying $\varphi$ and $q\curle_{nS} p_i$ for both $i=1,2$. We then show that the latter statement is equivalent to $\varphi$ being characteristic within $\mathL_{nS}$; that is, there is a process $q$ satisfying $\varphi$ such that for all processes $p$ satisfying $\varphi$, 
$q\curle_{nS} p$.

The game is presented in Table~\ref{tab:n-prime-game}. $B$ constructs two labelled trees, denoted  $T_1=(S_1,L_1,R^1_{a_1},\dots,R^1_{a_k})$ and $T_2=(S_2,L_2,R^2_{a_1},\dots,R^2_{a_k})$, and  $A$ constructs a third labelled tree denoted  $T_3=(S_3,L_3,R^3_{a_1},\dots,R^3_{a_k})$. The game starts with $S_1=\{p^1_0\}$, $S_2=\{p^2_0\}$, $S_3=\{q_0\}$, $L_1(p^1_0)=L_2(p^2_0)=L_3(q_0)=\{\varphi\}$, and all $R^i_{a_j}$ being empty, where $i=1,2,3$. We describe the $l$-th round of the game for $l\geq 1$. States $p_1,p_2$, and $q$ are  equal to $p^1_0,p^2_0$, and $q_0$, respectively, if $l\in\{1,2\}$ or $p_1,p_2$ are the states that $A$ chose at the end of round $l-1$ and $q$ is the state that $B$ chose at the end of round $l-1$, if $l>2$.
 \begin{table}
\begin{center}
\begin{tabular}{ | p{15cm} | } 
\hline
\textbf{$\mathbf{1}^{\text{st}}$ round.} 
\begin{enumerate}
    \item $A$ and $B$ play $\nmosimab(\{\varphi\},\{\varphi\})$. If $A$ wins, they continue playing the \nsimpre game. Otherwise, $B$ wins the \nsimpre game. 
    \item  $B$ plays moves B($\wedge$) and B($\vee$) on $p_i$, for both $i=1,2$, until no formula can be replaced in $L_i(p_i)$.
    If $\bigwedge L_i(p_i)$ becomes unsatisfiable, then $B$ loses.
    \item $A$ plays move A(sub) once on $q$ and then she plays moves A($\wedge$) and A($\vee$) on $q$ until no formula can be replaced in $L_3(q)$. 
    If $\bigwedge L_3(q)$ becomes unsatisfiable, then $A$ loses.
\end{enumerate} 

\textbf{$\mathbf{l}^{\text{th}}$ round, $\mathbf{l\geq 2}$.}
\begin{enumerate}
    \item For every $a_j\in \act$ and both $i=1,2$, $B$ plays as follows. He plays move B($\Diamond$) on $p_i$. Then, $B$ plays moves B($\wedge$) and B($\vee$)  on every $p_i'$ such that $(p_i,p_i')\in R_{a_j}^i$ until no formula can be replaced in $L_i(p_i')$.
    If, for some $s\in S_i$, $\bigwedge L_i(s)$ is unsatisfiable, then $B$ loses.
    \item  For every $a_j\in \act$, $A$ plays as follows. She plays move A($\Diamond$) on $q$. Then, for every $j$-successor $q'$ of $q$, $A$ plays move A(sub) once and moves A($\wedge$) and A($\vee$) on $q'$ until no formula can be replaced in $L_3(q')$. Finally, $A$ plays move A(rem) on $q$. If, for some $s\in S_3$, $L_3(s)$ is unsatisfiable, then $A$ loses.
    \item $B$ chooses a $j\in\{1,\dots,k\}$ and a $j$-successor $q'$ of $q$. If $q$ has no $j$-successors, then $B$ loses.
    \item $A$ chooses two states $p_1'$ and $p_2'$ that are  $j$-successors of $p_1$ and  $p_2$ respectively. If some of $p_1$ and $p_2$ has no $j$-successors, then $A$ loses.
     \item $A$ and $B$ play the following four games: (i) $\nmosimab(L_3(q'),L_1(p_1'))$, (ii) $\nmosimab(L_1(p_1'),L_3(q'))$, (iii) $\nmosimab(L_3(q'),L_2(p_2'))$, and (iv) $\nmosimab(L_2(p_2'),L_3(q'))$. If $A$ loses any of (i)--(iv), then $A$ loses.
    \item If $l=\md(\varphi)+2$, the game ends and $B$ wins. If $l\leq \md(\varphi)+1$, the $l+1$-th round  starts on $p_1'$, $p_2'$, and $q'$. 
\end{enumerate} \\
 \hline
\end{tabular}
\end{center}
\caption{The \nsimpre game, where $n\geq 3$, initiated on a satisfiable $\varphi\in\mathL_{nS}$.}
\label{tab:n-prime-game}
\end{table}

All moves in the \nsimpre game align with those in the \nsimeq game, except for the A(sub) and A(rem) moves. Below, we provide the intuition behind these two specific moves. Let $\act=\{a,b\}$ and consider a formula $\varphi$ of the form $\langle a\rangle\psi_1\wedge\langle a\rangle \psi_2\wedge [a]\psi\wedge [b]\ff$, which is characteristic within  $\mathL_{nS}$. Assume that $\psi_1\wedge\psi$ is not characteristic within $\mathL_{nS}$. Then, $\psi_2\wedge\psi$ must be characteristic and must entail $\psi_1\wedge \psi$, i.e.\ $\psi_2\wedge\psi\models\psi_1\wedge\psi$. This implies that removing $\langle a\rangle\psi_1$ from $\varphi$ yields a logically equivalent formula. Now, let $s\in S_3$ be 
a state in the tree constructed by $A$, such that $\finlabelthree(s)=\{\varphi\}$. When $A$ generates two states $s_1$ and $s_2$ with $L_3(s_i)=\{\psi_i,\psi\}$ for $i=1,2$, she can choose to apply move A(sub) to add all required formulae to $L_3(s_2)$, so that, in the end, $\finlabelthree(s_1)\subseteq \finlabelthree(s_2)$ holds. Thus, when $A$ plays A(rem), she can remove $s_1$. Furthermore, she can ensure that $\bigwedge\finlabelthree(s_2)$ remains characteristic. By applying A(sub) and A(rem) according to this strategy, $B$ is forced, at step 3, to choose a state whose label set corresponds to a characteristic formula. This, in turn, allows $A$ to complete each round without losing.

\begin{restatable}{proposition}{Awinsprime}\label{prop:A-wins-prime}
Let $\varphi\in\mathL_{nS}$, where $n\geq 3$, be satisfiable. Then, $A$ has a winning strategy for the \nsimpre game on $\varphi\in\mathL_{nS}$ iff $\varphi$ is characteristic for some process within $\mathL_{nS}$.
\end{restatable}
\begin{proofsketch}
   Let the \nsimpre game be initiated on $\varphi$. If $\varphi$ is characteristic within $\mathL_{nS}$, then $A$ has a strategy such that, for every $l\geq 2$, at the beginning of round $l$, $\bigwedge \finlabelthree(q)$  is characteristic within  $\mathL_{nS}$ and  $\bigwedge \finlabeli(p_i)\models \bigwedge \finlabelthree(q)$, for both $i=1,2$. She applies moves A(sub), A($\wedge$), A($\vee$) and A(rem) so that, for every $s\in S_3$, if $s$ is not removed from $T_3$, then $\bigwedge L_3(s)$ is characteristic and, for every choice of $B$ at step 3, she can respond with $p_1$ and $p_2$ such that $\bigwedge\finlabeli(p_i)\models\bigwedge \finlabelthree(q)$ for $i=1,2$. Therefore, she can continue playing the game until $B$ loses in some round. Moreover, in this case, the game does not last for more than $\md(\varphi)+1$ rounds. For the converse, assume that $A$ has a winning strategy for the \nsimpre game on $\varphi$. Let $r_1,r_2$ be two processes that satisfy $\varphi$ and let $B$ play $T_i$ consistently on $r_i$ for $i=1,2$. Then, there is a process $t$ that satisfies $\varphi$ such that $t\curle_{nS} r_i$, for $i=1,2$, and $|t|\leq (2m+1)^{m+1}$, where $m=|\varphi|$. In fact, $t$ is the process on which the strategy of $A$ is based. Then, we can show that there is a process $q$ that satisfies $\varphi$ such that $|q|\leq (2m+1)^{m+1}$ and $q\curle_{nS} r$, for every process $r$ that satisfies $\varphi$. Consequently, $\varphi$ is characteristic for some process within $\mathL_{nS}$.
\end{proofsketch}

\begin{restatable}{theorem}{nprimepspace}\label{prop:3S-char-in-pspace}
 The Formula Primality Problem for $\mathL_{nS}$, $n\geq 3$, is  \pspace-complete.
\end{restatable}

\begin{restatable}{corollary}{ncharpspace}\label{cor:3s-char-pspacec}
  Let $|\act|>1$.  Deciding whether a formula in $\mathL_{nS}$, $n\geq 3$, is characteristic for a process within $\mathL_{nS}$ is \pspace-complete.
\end{restatable}

\section{The complexity of deciding characteristic formulae within $\mathL_{2S}$}\label{section:deciding-char-2S}
The goal of this subsection is to establish the following results on the complexity of the Formula Primality problem for $\mathL_{2S}$ and of the problem of deciding characteristic formulae within $\mathL_{2S}$ (and modulo $\equiv_{2S}$). 
\begin{algorithm}
\caption{Algorithm $\mathtt{ConPro}$ that takes as input $\varphi\in\mathL_{2S}$, and extends the tableau construction for $\varphi$ with lines 20--30. }\label{alg:process-construction}
\DontPrintSemicolon
\KwIn{$\varphi\in\mathL_{2S}$}
{$S\gets\{s_0\}$}\;
{$L(s_0)=\{\varphi\}$, $d(s_0)\gets 0$}\;
{$BoxCount\gets 0$}\;
\lFor{all $a_j\in \act$}{$R_{a_j}\gets\emptyset$}
{$Q.\mathtt{enqueue}(s_0)$}\;
\While{$Q$ is not empty}{
{$s\gets Q.\mathtt{dequeue()}$}\;
\While{$L(s)$ contains $\psi_1\wedge\psi_2$ or $\phi_1\vee\phi_2$}
{{$L(s)\gets L(s)\setminus\{\psi_1\wedge\psi_2\}\cup\{\psi_1\}\cup\{\psi_2\}$}\;
{non-deterministically choose $\phi$ between $\phi_1$ and $\phi_2$}\;
{$L(s)\gets L(s)\setminus\{\phi_1\vee\phi_2\}\cup\{\phi\}$}\;}
\lIf{$\ff\in L(s)$}{stop}
\For{all $\langle a_j\rangle\psi\in L(s)$}
{{$S\gets S\cup\{s'\}$} \Comment{$s'$ is a fresh state}\;
{$L(s')=\{\psi\}\cup\{\phi\mid [a_j]\phi\in L(s)\}$}\;
{$d(s')\gets d(s)+1$}\;
{$R_{a_j}\gets R_{a_j}\cup\{(s,s')\}$}\;
{\lIf{$\ff\in L(s')$}{stop}}
{\lIf{$d(s')<\md(\varphi)+1$}{$Q.\mathtt{enqueue}(s')$}
}}
{Non-deterministically choose to go to line 6 or line 21}\;
{Non-deterministically choose $N\in\{1,\dots,|\varphi|-BoxCount\}$}\;
\For{$i\gets 1$  to $N$}{
{Non-deterministically choose $j\in\{1,\dots,k\}$}\;
{$S\gets S\cup\{s'\}$} \Comment{$s'$ is a fresh state}\;
{$L(s')=\{\phi\mid [a_j]\phi\in L(s)\}$}\;
{$d(s')\gets d(s)+1$}\;
{$R_{a_j}\gets R_{a_j}\cup\{(s,s')\}$}\;
\lIf{$\ff\in L(s')$}{stop}
{\lIf{$d(s')<\md(\varphi)+1$}{$Q.\mathtt{enqueue}(s')$}}
{$BoxCount\gets BoxCount+1$}
}}
{Return $S, R_{a_1},\dots, R_{a_k}$}
\end{algorithm}

\begin{algorithm}
\caption{Algorithm $\mathtt{Prime_{2S}}$ decides whether $\varphi\in\mathL_{2S}$ is prime. State $s_0^i$, $i=1,2$, denotes the first state that is added to $S_i$ by $\mathtt{ConPro}(\varphi)$.
Procedure 
$\mathtt{Process}(S,s_0^i,R_{a_1},\dots,R_{a_k})$ 
computes a process corresponding to the output of $\mathtt{ConPro}$ and 
$\mathtt{MLB}(p_1,p_2)$
returns the $\grcd_{\curle_{2S}}(p_1,p_2)$.}\label{alg:prime-twosim}
\DontPrintSemicolon
\KwIn{$\varphi\in\mathL_{2S}$}
{$(S_1,R^1_{a_1},\dots, R^1_{a_k})\gets \mathtt{ConPro}(\varphi)$}\;
{$p_1\gets \process{$S_1,s_0^{1}, R^1_{a_1},\dots, R^1_{a_k}$}$}\;
{$(S_2,R^2_{a_1},\dots, R^2_{a_k})\gets \mathtt{ConPro}(\varphi)$}\;
{$p_2\gets \process{$S_2,s_0^{2}, R^2_{a_1},\dots, R^2_{a_k}$}$}\;
\lIf{some of the two calls of $\mathtt{ConPro}(\varphi)$ stops without an output}{accept}
\Else{
{$g\gets\gcd{$p_1,p_2$}$}\;
\lIf{$g$ is empty}{reject}
\lIf{$g\models\varphi$}{accept}
\lElse{reject}
}
\end{algorithm}

\begin{restatable}{theorem}{twosprimeconp}\label{prop:2S-primality-conpc}
The Formula Primality problem for $\mathL_{2S}$ is \conp-complete.
\end{restatable}

\begin{restatable}{corollary}{twoschardp}\label{cor:2S-char-dp}
  Let $|\act|>1$.  Deciding whether a formula in $\mathL_{2S}$ is characteristic for a process within $\mathL_{2S}$, or modulo $\equiv_{2S}$, is 
  in \dpc.
\end{restatable}



In the case of $\mathL_{2S}$, the upper bound for the Formula Primality problem decreases from \pspace to \conp. This is mainly because, for a satisfiable formula $\varphi\in\mathL_{2S}$, there is always a tableau for $\varphi$---and so a corresponding process satisfying $\varphi$---of polynomial size~\cite{AcetoACI25}. Regarding the satisfiability problem for the logic, an execution of the standard non-deterministic tableau construction~\cite{HalpernM92} must result in a tableau for $\varphi$
(and a corresponding process that satisfies $\varphi$)
and, therefore, we obtain an \NP algorithm~\cite{AcetoACI25}. In contrast, for the Formula Primality problem, we accept the formula $\varphi$ under the following conditions: (i) all executions of the non-deterministic tableau construction fail---implying that $\varphi$ is unsatisfiable and hence prime;  or (ii) we run the tableau construction twice in parallel, and for each pair of executions that return two tableaux for $\varphi$, corresponding to two processes $p_1,p_2$ satisfying $\varphi$, we check whether there is a process $q$ that also satisfies $\varphi$ and is 2-nested-simulated by both $p_1$ and $p_2$. Note that a winning strategy for $A$ in the \nsimpre game on $\varphi$ is equivalent to the second condition for some $\varphi\in\mathL_{nS}$, $n\geq 3$. When we implement the procedure outlined above---see Algorithm~\ref{alg:prime-twosim}---each execution runs in polynomial time and, since we universally quantify over all such executions, the problem lies in \conp. 



We introduce two algorithms, namely $\mathtt{ConPro}$ in Algorithm~\ref{alg:process-construction}, and $\mathtt{Prime_{2S}}$ in Algorithm~\ref{alg:prime-twosim}. Let $\varphi\in\mathL_{2S}$ be an input to the first algorithm. Lines 1--19 of $\mathtt{ConPro}$ are an implementation of the tableau construction for $\varphi$---see~\cite{HalpernM92}. 
If $\varphi$ is unsatisfiable, then $\mathtt{ConPro}(\varphi)$ stops without returning an output because it stops at lines 12 or 18. In the case that $\mathtt{ConPro}(\varphi)$ returns an output, then its output is an LTS corresponding to a process that satisfies $\varphi$. If there are $r_1,r_2$ satisfying $\varphi$ such that $r_1\not\curle_S r_2$, the tableau construction cannot guarantee the generation of two processes that are not simulation equivalent. This is precisely the role of lines 20--30 in $\mathtt{ConPro}$. Given such processes $r_1, r_2$, when run twice, the algorithm can choose two processes $p_1,p_2$ based on $r_1,r_2$. During construction of $p_1$, lines 20--30 can be used to add  to $p_1$ up to $|\varphi|$ states that witness the failure of $r_1\curle_S r_2$. Note that in the case of the \simequiv game, player $B$ could follow a similar strategy by using move B($\square$) and introducing a trace that witnesses $r_1\not\curle_S r_2$. In the case of $\mathL_{2S}$, since the full tableau is constructed, the algorithm needs only to construct a ``small'' process that serves as a witness to the same fact. 

Algorithm $\mathtt{Prime_{2S}}$ decides whether its input $\varphi\in\mathL_{2S}$ is prime: $\varphi$ is prime iff every execution of $\mathtt{Prime_{2S}}(\varphi)$ accepts. This algorithm runs $\mathtt{ConPro}(\varphi)$ twice. If $\mathtt{ConPro}(\varphi)$ fails to return an output, $\mathtt{Prime_{2S}}(\varphi)$  rejects at line 5---this line deals with unsatisfiability. For every two processes $p_1,p_2$ that satisfy $\varphi$, at line 7, $\mathtt{Prime_{2S}}(\varphi)$ constructs their maximal lower bound, denoted $\grcd_{\curle_{2S}}(p_1,p_2)$, which is a process $g$ that is 2-nested-simulated by both $p_i$'s and $r\curle_{2S} g$, for every process $r$ such that $r\curle_{2S} p_i$. Processes $p_1,p_2$ have a maximal lower bound iff $p_1\equiv_S p_2$.
In case $\grcd_{\curle_{2S}}(p_1,p_2)$ does not exist, the algorithm discovers two processes satisfying $\varphi$ such that there is no process that is 2-nested-simulated by both of them and it rejects the input---$\varphi$ is not prime. On the other hand, if $\grcd_{\curle_{2S}}(p_1,p_2)$ exists, then there is a process that is 2-nested simulated by both $p_i$'s and it can be constructed in polynomial time. It remains to check whether $\grcd_{\curle_{2S}}(p_1,p_2)$ satisfies $\varphi$. If so, then the second condition described above is met, and the algorithm accepts. If $\grcd_{\curle_{2S}}(p_1,p_2)\not\models\varphi$ it can be shown that there is no process $r$ satisfying $\varphi$ that is 2-nested-simulated by both $p_i$'s, and the algorithm rejects at line 10. This establishes Theorem~\ref{prop:2S-primality-conpc}. By slightly adjusting  $\mathtt{Prime_{2S}}$, we can show that deciding whether all processes satisfying a formula in $\mathL_{2S}$ are 2-nested-simulation equivalent is in \conp. By these results and the \NP-completeness of the satisfiability problem for $\mathL_{2S}$, we obtain the upper bound given in Corollary~\ref{cor:2S-char-dp}.
\bibliographystyle{eptcs}
\bibliography{bibliography}

\end{document}